\documentclass[smallextended]{svjour3}       
\smartqed  

\usepackage{graphicx}
\usepackage{amsmath,amssymb}
\usepackage{amssymb}
\usepackage{amsmath}

\newtheorem{condition}{Condition}

\begin{document}

\title{Cyber Insurance}



\author{Quanyan Zhu}
%
\institute{New York University, Brooklyn, NY, 11201; Email: qz494@nyu.edu}
%
\date{}
%

\maketitle

\section{Introduction}
\label{intro}

Critical Infrastructures are increasingly dependent on the information and communication technologies (ICTs) to sense, transmit, and fuse data for real-time operations and control of the infrastructures. The heavy integration of the ICTs has also brought many potential threat that can cause data privacy breaches, availability of the services, and the cascading damages. The vulnerabilities in the ICT can arise not only from unintentional misconfiguration and mismanagement of the protocols and devices but also from the intentional injection of the malware, spread of the worms, and the cyber attacks. Recent cyber attacks on Iranian nuclear power plants and the Ukraine power grid have also shown that the attacks are becoming increasingly sophisticated. For example, advanced persistent threats (APTs) \cite{Samuel2010,McMillan2010}, such as Stuxnet, Flame, and Duqu, can exploit zero-day vulnerabilities, leverage human errors and insider threats, and move stealthily in the network before launching a successful attack. These attacks are often difficult to detect and prevent as they have access to a sufficient amount of resources and are capable of staying in the victim's system for years. 

Hence, cyber risks for infrastructures are of growing concerns. The cyber risk will not only create cyber incidents including identity theft, cyber extortion, and network disruption, but also can lead to the malfunction of the entire infrastructure and its key services to users and customers.  It becomes a critical issue for operators to safeguard infrastructures from the intentional and unintentional actions that would inflict damage on the system.   Conventional countermeasures include installing intrusion detections, blacklisting malicious hosts, filtering/blocking traffic into the network.  However, these methods cannot guarantee perfect security and can be evaded by sophisticated adversaries despite the advances in technologies. Therefore, cyber risks are inevitable and it is essential to find other means to mitigate the risks and impact. 

Cyber insurance is an important tool in risk management to transfer risks \cite{zhang2017bi,zhang2019flipin}. Complement to the technological solutions to cybersecurity, cyber insurance can mitigate the loss of the targeted system and increase the resiliency of the victim by enabling quick financial and system recovery from cyber incidents. Such scheme is particularly helpful to small and medium size infrastructure systems that cannot afford a significant investment in cyber protection. The market of cyber insurance is still in its infancy. U.S. penetration level of the insured is less than 15\%. Promisingly, the market is growing fast at a 30\% annual growth rate since 2011.  The key challenge with cyber insurance lies in the difficulty to assess different types of cyber risks and impact of the cyber incidents that are launched by resourceful adversaries who are stealthy and purposeful. The design of cyber insurance also needs to take into account moral hazards and adverse selection problems. The insured tend to lack in incentives to improve their security measures to safeguard against attacks. As the nodes in the cyber space are increasingly connected, the unprotected cyber risks can propagate to other uninsured nodes. With asymmetric information of the insured, the insurer also has tendency to increase the premium rates for higher risks, making the cyber insurance less affordable to end users. 

In this chapter, we aim to provide a baseline framework to understand the interactions among the players in the cyber insurance market and leverage it to design optimal cyber insurance for the infrastructure services. One key application of the framework is to provide assurance to the infrastructure managers and users and transfer their risks when the attack on power grid fails to provide electric power to a food processing plant, when cloud servers break down and fail to provide airline customer check-in information, and when the trains collide due to the communication systems fail. In the examples above, it is clear that the cyber insurance play a key role in mitigating the cyber risks that interconnect the communications and information systems of an infrastructure with their physical impact on the infrastructure or linked infrastructures. The interdependencies among the infrastructures and their operators and users can propagate the cyber risks and exacerbate the damages on the critical infrastructures. To this end, the understanding of the cyber insurance of interconnected players is the key to the holistic understanding of the risk management of interdependent infrastructures.

This chapter will first present a principal-agent game-theoretic model to capture the interactions between one insurer and one user. The insurer is deemed as the principal who does not have incomplete information about user's security policies. The user, which refers to the infrastructure operator or the customer, implements his local protection and pays a premium to the insurer. The insurer designs an incentive compatible insurance mechanism that includes the premium and the coverage policy, while the user determines whether to participate in the insurance and his effort to defend against attacks. The chapter will also focus on an attack-aware cyber insurance model by introducing the adversarial behaviors into the framework. The behavior of an attacker determines the type of cyber threats, e.g. denial of service (DoS) attacks, data breaches, phishing and spoofing. The distinction of threat types plays a role in determining the type of losses and the coverage policies. The data breaches can lead to not only financial losses but also damage of the reputations. The coverage may only cover certain agreed percentage of the financial losses. 

\section{Background}

The challenges of cyber security are not only technical issues but also economic and policy issues \cite{anderson2006economics}. Recently, the use of cyber insurance to enhance the level of security in cyber-physical systems has been studied \cite{kesan2005cyberinsurance,lelarge2008local}. While these works deal with externality effects of cyber security in networks, few of them take into account in the model the cyber attack from a malicious adversary to distinguish from classical insurance models. In \cite{pal2014will}, the authors have considered direct and indirect losses, respectively due to cyber attacks and indirect infections from other nodes in the network. However, the cyber attacks are taken as random inputs rather than a strategic adversary. The moral hazard model in economics literature \cite{holmstrom1979moral,holmstrom1982moral} deal with hidden actions from an agent, and aims to address the question: How does a principal design the agent's wage contract to maximize his effort? This framework is related to insurance markets and has been used to model cyber insurance \cite{bolot2009cyber} as a solution for mitigating losses from cyber attacks. In addition, in \cite{acemoglu2013network}, the authors have studied a security investment problem in a network with externality effect. Each node determines his security investment level and competes with a strategic attacker. Their model does not focus on the insurance policies and hidden-action framework. In this work, we enrich the moral-hazard type of economic frameworks by incorporating attack models, and provide a holistic viewpoint towards cyber insurance and a systematic approach to design insurance policies. The network effect on security decision process has been studied in \cite{miura2008security}. The authors have considered a variation of the linear influence networks model in which each node represents a network company and directed links model the positive or negative influence between neighbor nodes.

Game-theoretic models are natural frameworks to capture the adversarial and defensive interactions between players \cite{zhu2018multi,rass2017physical,zhuang2010modeling,miao2018hybrid,farhang2014dynamic,manshaei2013game,zhu2013deployment,zhang2017strategic,horak2017manipulating,huang2017large}.    Game theory can provide a quantitative measure of the quality of protection with the concept of Nash equilibrium where both defender and an attacker seek optimal strategies, and no one has an incentive to deviate unilaterally from their equilibrium strategies despite their conflict for security objectives. The equilibrium concept also provides a quantitative prediction of the security outcomes of the scenario the game model captures. There are various types of game models that can capture different class of adversaries. For example, games of incomplete information are suitable for understanding cyber deception \cite{pawlick2018modeling,zhang2017strategic,horak2017manipulating,pawlick2017game,pawlick2015deception,zhuang2010modeling}; dynamic games are useful for modeling cyber-physical system security \cite{xu2017secure,xu_game-theoretic_2017,xu_cross-layer_2016,xu2015cyber,huang2017large,chen2017dynamic,miao2018hybrid,yuan2013resilient}; and zerosum and Stackelberg games for security risk management  \cite{zhang2018game,zhang2015secure,wang2017detection,zhang2017game,pawlick_stackelberg_2016,pawlick_mean-field_2017}. In this work, we build a zerosum game between an attacker and a defender to quantify the cybersecurity risks associated with adversarial behaviors. This game model is nested in the principle-agent game models to establish an integrated framework that captures the defender, the attacker, and the insurer.

\section{Three-Person Game Framework for Cyber Insurance}

In this section, we introduce a principal-agent model for cyber insurance that involve users and insurers.
Users here can refer to an infrastructure operator that manages cyber networks that face threats from an attacker, making users vulnerable to data breaches, task failures, and severe financial losses. The objective of the users is to find an efficient way to mitigate the loss due to the cyber attacks. To this end, there are two main approaches. One is to deploy local protections, such as firewalls and intrusion detection systems (IDSs) \cite{raiyn2014survey,axelsson2000intrusion}, frequent change of passwords, timely software patching and proactive moving target defenses \cite{jajodia2011moving}. These defense mechanisms can reduce the success rate of the attacks, but cannot guarantee perfect network security for users.  There are still chances for the users to be hacked by the attackers.  The other approach is to adopt cyber-insurance. The users pay a premium fee so that the loss due to cyber attacks can be compensated by the insurer. This mechanism provides an additional layer of mitigation to reduce the loss further that the technical solutions of the first approach cannot prevent. To capture the two options in our framework, we allow users to decide their protection levels as well as their rational choice of participation in the insurance program.
 
 Attackers are the adversaries who launch cyber-attacks, such as node capture attacks\cite{tague2008modeling}  and denial of services (DoS) attacks \cite{jhaveri2012attacks}, to acquire private data from users or cause disruptions of the network services.  Hence, the objective of the attacker is to find an efficient attack strategy to inflict as much damage to the users as possible. We use attack levels to represent different attack strategies to capture various types of attacks of different levels of severity. 
A higher attack level is more costly to launch, but it will create more severe damage. Since the loss of the users not only depends on the attack strategies but also insurance policies. The optimal strategy of the attacker will also be influenced by the coverage levels of an insurance policy. 

An insurer is a person or company that underwrites an insurance risk by providing users an incentive compatible cyber-insurance policy that includes a premium and the level of coverage. The premium is a subscription fee that is paid by the users to participate in the insurance program while the coverage level is the proportion of loss that will be compensated by the insurer as a consequence of successful cyber attacks. The insurers have two objectives. One is to make a profit from providing the insurance, and the other one is to reduce the average losses of the users, which is also directly related to the cost of the insurer. An insurer's problem is to determine the subscription fee and the coverage levels of the insurance. Note that the average losses depend on both users' local protection levels and attackers' attack levels. Moreover, the rational users will only enroll in the insurance when the average reduction in the cost is higher than or equal to the premium he paid to the insurer.
%

The objectives of users, attackers, and insurers, and the effects of their actions are all intertwined. We use a 3-player game to capture the complex interactions among the three parties. The conflicting objectives of a user and an attacker can be captured by a local game at each node in which the user determines a defense strategy while the adversary chooses an attack strategy. The outcome of the local interactions at each node determines its cyber risk and the cyber insurance is used as an additional method to further reduce the loss due to the cyber risk. The insurers are the leaders or principals in the framework who design insurance policies for the users while the users can be viewed as followers or agents who determine their defense strategies under a given insurance policy. 

\begin{figure}[]
\centering
\includegraphics[width=0.8\textwidth]{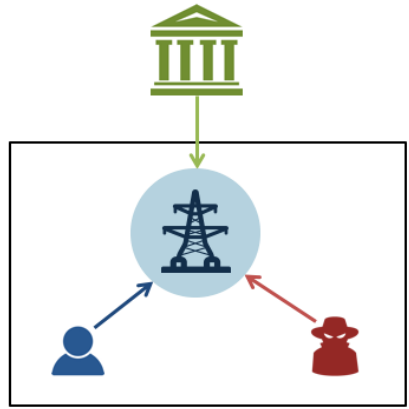}
\vspace{-25mm}
\caption{Three-player game among user, attacker, and the insurer.}
\label{3playergame}
\end{figure}

\subsection{Attack-Aware Cyber Insurance}

We first formulate the game between the user and the attacker, then we describe the insurer's problem under the equilibrium of the user and the attacker's game. An illustration of the cyber-insurance model is shown in Fig. \ref{fig:1U1A1IExample}. 
\begin{figure}[]
\centering
\includegraphics[width=0.8\textwidth]{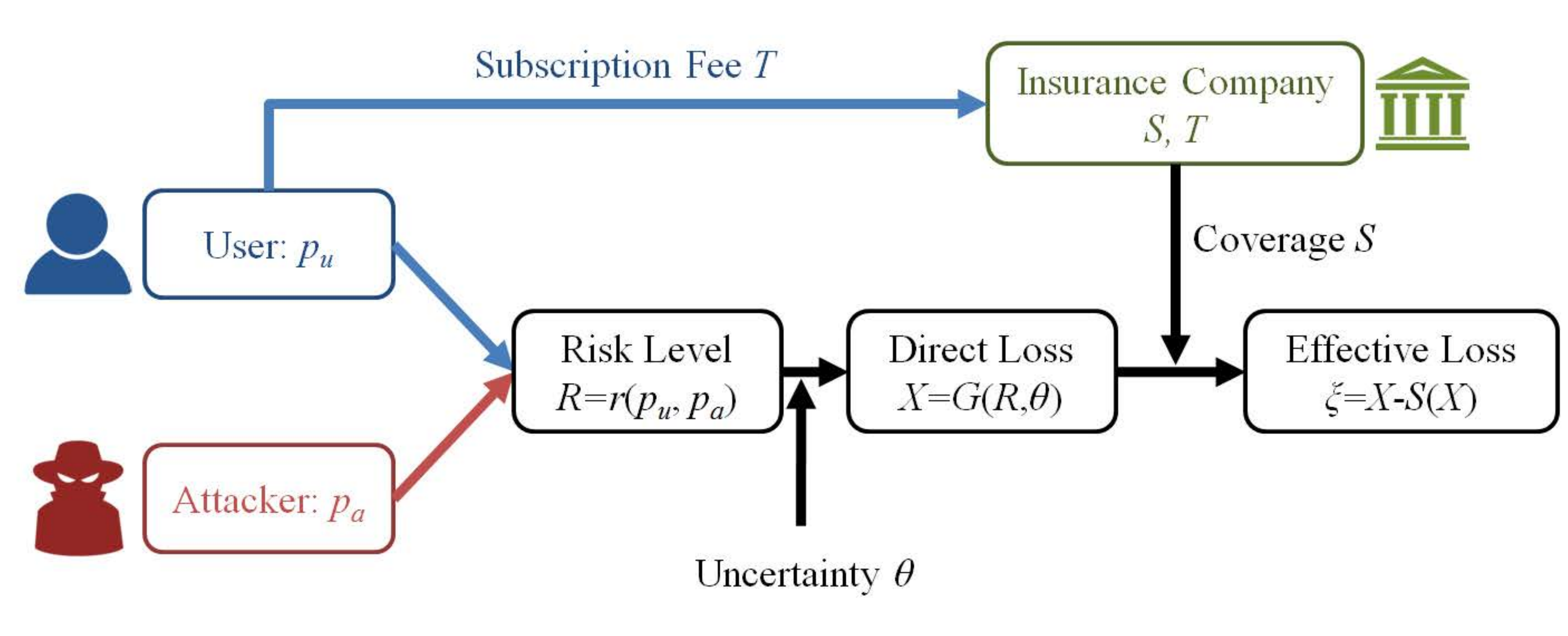}
\caption{Illustration of the interactions between three players: The action pair $(p_u, p_a)$ chosen by the user and the attacker results in a risk level not directly observable by the insurer. The insurer designs an insurance policy that includes a premium subscription fee and the coverage level to cover part of the loss due to the cyber attack.}
\label{fig:1U1A1IExample}
\end{figure}
Let $p_u \in [0,1]$ and $p_a \in [0,1]$ denote the local protection level of the user and the attack level of the attacker. On one hand, a large $p_u$ indicates a cautious user while a small $p_u$ indicates that the user is reckless. A reckless user may click on suspicious links of received spam emails, fail to patch the computer system frequently, and leave cyber footprints for an adversary to acquire system information. On the other hand, a large $p_a$ indicates a powerful attacker, and a small $p_a$ indicates a powerless attacker. The abstraction of using $p_u$ and $p_a$ captures the effectiveness of a wide range of heterogeneous defense and attack strategies without a fine-grained modeling of individual mechanisms. This will allow us to focus on the consequence of security issues and the choice of a mechanism that induces the result.

The action pair of the user and the attacker $(p_u,p_a)$ determines the risk level of the user $R\in \mathbb{R}_{\geq 0}$. A larger $p_u$ and a smaller $p_a$ indicate a higher risk level of the user. We use the following risk function $r$ to denote the connections between the user's and the attacker's actions and the risk level of the user.  
\begin{definition}
\label{def:RiskLevel}
Function $r(p_u,p_a):[0,1]^2\rightarrow \mathbb{R}_{\geq 0}$ gives the risk level $R$ of the user with respect to the user's local protection level $p_u$ and the attack's attack level $p_a$. Moreover, it is assumed to be continuous on $(0,1]^2$, convex and monotonically decreasing on $p_u \in [0,1]$, and concave and monotonically increasing in $p_a \in [0,1]$. 
\end{definition}
Note that the monotonicity in $p_u \in [0,1]$ indicates that a larger local protection level of user leads to a smaller risk level of user while the monotonicity in $p_a \in [0,1]$ indicates that a larger attack level of attacker leads to a larger risk level of user. Since $r$ is convex on $p_u$, the risk decreases smaller when the user adopts larger local protection level. Since $r$ is concave on $p_a$, the risk increases faster when the attacker conducts a higher attack level. Without loss of generality, we use the following risk function,
\begin{equation}
\label{eq:Risk}
r(p_u,p_a)= \log\left(\frac{p_a}{p_u}+1\right).
\end{equation}
Similar types of functions have also been widely used in jamming attacks in wireless networks \cite{manshaei2013game,altman2007jamming} and rate control problems \cite{kelly1998rate,zhu2012guidex}. Under the risk level of $R$, the economic loss of the user can be represented as a random variable $X$ measured in dollars, which can be expressed as $X={G}(R,\theta)$, where $\theta$ is a random variable with probability density function $g$ that captures the uncertainties in the measurement or system parameters. For example, a data breach due to the compromise of a server can be a consequence of low security level at the user end. The magnitude of the loss depends on the content and the significance of the data, and the extent of the breach. The variations in these parameters are captured by the random variable $\theta$. Since the risks of being attacked cannot be perfectly eliminated, the user can transfer the remaining risks to the third party, the insurer, by paying a premium or subscription fee $T$ for a coverage of $S(X)$ when he faces a loss of $X$, where $S:\mathbb{R}_{\geq 0}\rightarrow\mathbb{R}_{\geq 0}$ is the payment function that reduces the loss of the user if he is insured. Thus, the effective loss $\xi$ to the user becomes $\xi =X-S(X)$.

Given the attacker's action $p_a$ and the insurer's coverage function $S$, the user aims to minimize the average effective loss by finding the optimal local protection level $p_u^*$. Such objective can be captured by the following optimization problem
\begin{equation}
\label{eq:UserMin}
\min\limits_{p_u \in [0,1]} \mathbb{E}[H(\xi)]=\mathbb{E}[H(X-S(X))],
\end{equation} 
where $H:\mathbb{R}_{\geq 0} \rightarrow \mathbb{R}_{\geq 0}$ is the loss function of the user, which is increasing on $\xi$. Note that the expectation is taken with respect to the statistics of $\theta$. The subscription fee $T$ is not included in this optimization problem, as the fee is a constant decided by the insurer. 

The loss function $H(\xi)$ indicates the user's risk propensity. A convex $H(\xi)$ indicates that the user is risk-averse, i.e., the user cares more about the risk, while a concave $H(\xi)$ indicates that the user is risk-taking, i.e., he cares more about the cost, rather than the risk. A linear $H(\xi)$ in $\xi$ indicates that the user is risk-neutral. In this paper, we consider a risk-averse user, and use a typical risk-averse loss function that $H(\xi) = e^{\gamma\xi}$ with $\gamma>0$, where $\gamma$ indicates how much the user cares about the loss. 

Note that the cost function in (\ref{eq:UserMin}) can be expressed explicitly as a function of $X$. Thus, Problem (\ref{eq:UserMin}) can be rewritten by taking expectations with respect to the sufficient statistics of $X$. Let $f$ be the probability density function of $X$. Clearly, $f$ is a transformation from the density function $g$ (associated with the random variable $\theta$) under the mapping ${G}$. In addition, $g$ also depends on the action pair $(p_u,p_a)$ through the risk variable $R$. Therefore, we can write $f(x|p_u,p_a)$ to capture the parameterization of the density function. Without loss of generality, we assume that $X$ follows an exponential distribution, i.e., $X \sim \exp(\frac{1}{R})$,  where $R:=r(p_u,p_a)$ is the risk level of the user. The exponential distribution has been widely used in risk and reliability analysis\cite{minkova2010insurance,finkelstein2008failure,balakrishnan1996exponential,christoffersen2004backtesting}. Thus the density function can be written as 
\begin{eqnarray}
\nonumber f(x|p_u,p_a) &=& \frac{1}{R}e^{-\frac{1}{R}x}=\frac{1}{r(p_u,p_a)}e^{-\frac{1}{r(p_u,p_a)}x}
 \\
\nonumber &=&\frac{1}{\log(\frac{p_a}{p_u}+1)}e^{-\frac{1}{\log(\frac{p_a}{p_u}+1)}x}, \forall x\in\mathbb{R}_{\geq 0}.
\end{eqnarray}
The average amount of loss given actions $p_u$ and $p_a$ is $\mathbb{E}(X) = R = r(p_u,p_a)=\log(\frac{p_a}{p_u}+1)$. For small $p_u$ and  large $p_a$, the risk level of the user $R$ tends to be large, which leads to a large average loss of the user. We further assume that the insurance policy $S(X)$ is linear in $X$, i.e., $S(X)=sX$, where $s\in[0,1]$ indicates the coverage level of the insurance. Hence, the effective loss is given by $\xi = (1-s)X$. The average effective loss given the insurance coverage level $s$ and the action pair $(p_u,p_a)$ is $\mathbb{E}(\xi)=\mathbb{E}((1-s)X) = (1-s)\mathbb{E}(X) =(1-s)\log(\frac{p_a}{p_u}+1)$. When $s$ is large, the effective loss is small. As a result,  we arrive at
\begin{equation}
\label{eq:UserH}
\begin{array}{l}
\mathbb{E}[H(\xi)]: = \displaystyle\int_{x_i\in \mathbb{R}_{\geq 0}} H(x-S(x))f(x|p_u,p_a)dx \\ \textrm{\ \ \ \ \ \ \ \ \ \ } =\frac{1}{R}\displaystyle\int_{0}^{\infty} e^{[\gamma(1-s)-\frac{1}{R}]  x}  dx\\ \textrm{\ \ \ \ \ \ \ \ \ \ } =\frac{1}{1-\gamma(1-s)R} 
\\ \textrm{\ \ \ \ \ \ \ \ \ \ }=  \frac{1}{1-\gamma(1-s)\log(\frac{p_a}{p_u}+1) } .
\end{array}
\end{equation}
The loss is finite when
\begin{equation}
\label{eq:feasible}
\gamma(1-s)-\frac{1}{R}< 0, \textrm{\ i.e., \ }   1 -  \gamma (1-s)\log(\frac{p_a}{p_u}+1) > 0.
\end{equation}
Otherwise, the loss will be infinite, i.e., $\mathbb{E}[H(\xi)]\rightarrow \infty$. In this regime, no insurance scheme can be found to mitigate the loss. Condition (\ref{eq:feasible}) gives a feasible set of parameters under which cyber insurance is effective and provides a fundamental limit on the level of mitigation. Note that minimizing (\ref{eq:UserH}) is equivalent as minimizing $\gamma(1-s)\log(\frac{p_a}{p_u}+1)$ under the feasible equality (\ref{eq:feasible}). The user's problem can be rewritten as follows:
\begin{equation}
\label{eq:UserMinX}
\begin{array}{l}
\min\limits_{p_u \in [0,1]} J_u(p_u,p_a,s):= \gamma(1-s)R=\gamma(1-s)\log(\frac{p_a}{p_u}+1)\\ \begin{array}{cc}
{\textrm{s.t.}}&{1 -  \gamma (1-s)\log(\frac{p_a}{p_u}+1) > 0.}
\end{array}
\end{array}
\end{equation} 
Problem (\ref{eq:UserMinX}) captures the user's objective to minimize average effective loss given the attack level $p_a$ and the insurance coverage level $s$. On the other hand, the attacker aims to find the optimal attack level $p_a^*$ that maximizes the average loss of the user given user's local protection level and insurer's coverage level $s$. Such conflicting interests of the user and the attacker constitutes a zero-sum game, which takes the following minimax or max-min form, 
\begin{equation}
\label{eq:MinMax}
\begin{array}{ccc}
{\begin{array}{l}
\min\limits_{p_u\in [0,1]}\max\limits_{p_a \in [0,1]} K(p_u,p_a,s)\\ \begin{array}{cc}
{\textrm{s.t.}}&{(p_u,p_a)\in \mathcal{S}_{u,a}(s).}
\end{array} 
\end{array}}&{\textrm{\ \ \ or \ \ \ }}&{\begin{array}{l}
\max\limits_{p_a \in [0,1]}\min\limits_{p_u\in [0,1]} K(p_u,p_a,s)\\ \begin{array}{cc}
{\textrm{s.t.}}&{(p_u,p_a)\in \mathcal{S}_{u,a}(s).}
\end{array} 
\end{array}}
\end{array}
\end{equation}
where 
\begin{equation}
\label{eq:MinMaxK}
K(p_u,p_a,s):=\gamma(1-s)R + c_u p_u -  c_a p_a= \gamma(1-s)\log(\frac{p_a}{p_u}+1) + c_u p_u -  c_a p_a,
\end{equation}
\begin{equation}
\label{eq:MinMaxS}
\mathcal{S}_{u,a}(s):=\left\lbrace (p_u,p_a) \Big| 1 -  \gamma (1-s)\log(\frac{p_a}{p_u}+1) > 0   \right\rbrace.
\end{equation}
The first term of the objective function $K$ captures the average effective loss given insurance coverage level $s$, the local protection level $p_u$ and the attack level $p_a$. The second and third terms indicate the cost of the user and the attacker, respectively. $c_u\in \mathbb{R}_{>0}$ is the cost parameter of the user. A larger $c_u$ indicates that local protection is costly. $c_a\in \mathbb{R}_{>0}$ denotes the cost parameter of the attacker to conduct an attack level of $p_a$. A larger $c_u$ indicates that a cyber-attack is costly. Note that $c_u$ and $c_a$ can be interpreted as the market price of local protections and cyber-attacks, and they are known by the insurer. The constraint indicates the feasible set of the user. Note that if $s$, $p_u$, and $p_a$ are not feasible, $K$ is taken to be an infinite cost. Minimizing  $K(p_u,p_a,s)$ captures the user's objective to minimize the average effective loss with the most cost-effective local protection level. Maximizing $K(p_u,p_a,s)$ captures the attacker's objective to maximize the average effective loss of the user with least attack level. Note that the minimax form of (\ref{eq:MinMax}) can be interpreted as a worst-case solution for a user who uses the best security strategies by anticipating the worst-case attack scenarios.

Furthermore, Problem (\ref{eq:MinMax}) yields a saddle-point equilibrium (SPE) to the insurance coverage level $s$ which can be defined as follows:
\begin{definition}
\label{def:SaddlePoint}
Let $\mathcal{S}_u(s)$, $\mathcal{S}_a(s)$ and $\mathcal{S}_{u,a}(s)$ be the action sets for the user and the attacker given an insurance coverage level $s$. Then the strategy pair $(p_u^*,p_a^*)$ is a saddle-point equilibrium (SPE) of the zero-sum game defined by the triple $G_z:=\langle \{User,Attacker\},\{\mathcal{S}_u(s),\mathcal{S}_a(s),\mathcal{S}_{u,a}(s)\},K\rangle$, if for all $p_u \in \mathcal{S}_u(s),p_a \in \mathcal{S}_a(s),(p_u,p_a)\in\mathcal{S}_{u,a}(s)$, 
\begin{equation}
\label{eq:DefSaddlePoint}
\begin{array}{cr}
{K(p_u^*,p_a,s)\leq  K(p_u^*,p_a^*,s) \leq  K(p_u,p_a^*,s),}&{}
\end{array}
\end{equation}
where $K$ and $\mathcal{S}_{u,a}(s)$ is the objective function and feasible set defined in (\ref{eq:MinMaxK}) and (\ref{eq:MinMaxS}).
\end{definition}
The definition indicates that if a pair $(p_u^*,p_a^*)$ satisfies (\ref{eq:DefSaddlePoint}), then it is a SPE of the game between the user and the attacker to the insurer's insurance policy. Note that under a given insurance coverage level $s$, $(p_u^*,p_a^*)$ must satisfy the feasible constraint (\ref{eq:feasible}). Thus, we aim to look for a constrained SPE of the zero-sum game with coupled constraints on the strategies of the players.
\begin{proposition}
\label{pro:SaddlePointSol}
Given an insurance coverage level $s$ that satisfies 
\begin{equation}
\label{eq:SaddlePointFea}
1 -  \gamma (1-s)\log\left(\frac{c_u}{c_a}+1\right) > 0,
\end{equation}
there exists a unique SPE of the zero-sum game defined in Definition \ref{def:SaddlePoint}, given by
\begin{equation}
\label{eq:SaddlePointSol}
\begin{array}{cc}
{p_u^* = \frac{\gamma (1-s)}{c_u+c_a},}&{p_a^* = \frac{c_u \gamma (1-s)}{c_a(c_u+c_a) }.}
\end{array}
\end{equation}
\end{proposition}
Proposition \ref{pro:SaddlePointSol} shows that the SPE of the zero-sum game between the user and the attacker is related to the insurer's policy $s$. Note that when $s$ is large, both the $p_u^*$ and $p_a^*$ is small, indicating that both the user and the attacker will take weak actions. Moreover, we have the following observations regarding the SPE.
\begin{remark}[Peltzman Effect]
\label{rem:Peltzman}
When the insure provides higher coverage level $s$, the SPE of the user $p_u^*$ tend to be smaller, i.e., the user takes a weaker local protection. Such risky behavior of the user in response to insurance is usually referred as Peltzman effect \cite{peltzman1975effects}. 
\end{remark}
\begin{corollary}[Invariability of The SPE Ratio]
\label{rem:attackeruser}
The SPE satisfies $
\frac{p_a^*}{p_u^*}  = \frac{c_u}{c_a}$, i.e., the ratio of the actions of the user and the attacker is only related to $c_u$ and $c_a$, and it is independent of the insurer's policy $s$. In particular, when $c_u = c_a$, $\frac{p_a^*}{p_u^*} = 1$, i.e., the SPE becomes symmetric, as $p_u^*=p_a^* = \frac{\gamma (1-s)}{c_u+c_a}=\frac{\gamma (1-s)}{2c_u}=\frac{\gamma (1-s)}{2c_a}$.
\end{corollary}
\begin{remark}[Constant Cost Determined SPE Risk]
\label{rem:SPRisk}
The user has a constant saddle-point risk level $R^*=r(p_u^*,p_a^*)=\log \left(\frac{p_a^*}{p_u^*}+1\right) =\log \left(\frac{c_u}{c_a}+1\right) $ at the equilibrium, which is determined by the costs of adopting protections and launching attacks. The ratio is independent of coverage level $s$.
\end{remark}
\begin{corollary}
\label{rem:Exloss}
At the saddle point, the average direct loss of the user is $
\mathbb{E}(X) = R^* = \log\left(\frac{c_u}{c_a}+1\right)$, the average effective loss of the user is $
\mathbb{E}(\xi) = \mathbb{E}((1-s)X) = (1-s)\mathbb{E}(X) =(1-s)R^*= (1-s)\log\left(\frac{c_u}{c_a}+1\right)$, the average payment of the insurer to the user is $\mathbb{E}(sX) = s\mathbb{E}(X) =sR^*= s\log\left(\frac{c_u}{c_a}+1\right)$.

\end{corollary}

Corollary \ref{rem:attackeruser} indicates the constant saddle-point strategy ratio of the user and the attacker, which is determined only by the cost parameters $c_u$ and $c_a$, i.e., the market prices or costs for applying certain levels of protections and attacks, respectively. As a result, the saddle-point risk level of the user is constant, and only determined by the market as shown in Remark \ref{rem:SPRisk}. Thus, the average direct loss is constant as shown in Corollary \ref{rem:Exloss}. However, when the insurance coverage level $s$ does not satisfy (\ref{eq:SaddlePointFea}), the insurability of the user is not guaranteed, which is shown in the following proposition.
\begin{proposition}[Fundamental Limits on Insurability]
\label{pro:Fun}
Given an insurance coverage level $s$ that $1-\gamma(1-s)\log\left(\frac{c_u}{c_a}+1\right) \leq 0$, $(p_u^*,p_a^*)$ does not satisfy the feasible inequality (\ref{eq:feasible}), thus, the average direct of the user $\mathbb{E}(X)\rightarrow \infty$, and the zero-sum game defined in Definition \ref{def:SaddlePoint} does not admit a SPE. Thus, the user is not insurable, as the insurance policy cannot mitigate his loss. The insurer will not also provide insurance to the user who is not insurable.
\end{proposition}
\begin{proposition}
\label{pro:FunIns}
Under an insurable scenario, the cost parameter of the user must satisfy $c_u< c_a (1-e^{-\gamma(1-s)})$, and the local protection level of the user must satisfy $p_u>\frac{\gamma(1-s)}{c_a}e^{\gamma(1-s)}$.
\end{proposition}
\begin{proof}
The first inequality can be easily achieved from (\ref{eq:SaddlePointFea}). From Appendix A, given the action of the user $p_u$, the best action of the attacker is $P_a^*(p_u) = \frac{\gamma(1-s)}{c_a}-p_u$. By plugging $P_a^*(p_u)$ into the feasible inequality (\ref{eq:feasible}), we can get $p_u > \frac{\gamma(1-s)}{c_a}e^{\gamma(1-s)}$.
\end{proof}
It is important to note that the user must pay a subscription fee $T\in\mathbb{R}_{\geq 0}$ to be insured. The incentive for the user to buy insurance exists when the average loss at equilibrium under the insurance is lower than the loss incurred without insurance. If the amount of the payment from the insurer is low, then the user tends not to be insured. In addition, if the payment is low, then the risk for the insurer will be high and the user may behave recklessly in the cyber-space.
\subsection{Insurer's Problem}
The insurer announces the insurance policy $\{s,T\}$, where $s$ indicates the coverage level, $T$ indicates the subscription, and then the user's and the attacker's conflicting interests formulates a zero-sum game, which yields a unique solution as shown in Proposition \ref{pro:SaddlePointSol}, with the corresponding equilibrium loss as shown in Corollary \ref{rem:Exloss}. Note that $T$ is the gross profit of the insurer as he charges it from the user first, but when the user faces a loss $X$, the insurer must pay $sX$ to the user. The operating profit of the insurer can be captured as $T-sX$. The insurer cannot directly observe the actions of the user and the attacker. However, with the knowledge of the market, i.e., the cost parameters of the user $c_u$ and the attacker $c_a$, the insurer aims to minimize the average effective loss of the user while maximizing his operating profit. 

Recall Corollary \ref{rem:Exloss}, the average effective loss of the user at saddle-point is $\mathbb{E}(\xi)=(1-s)\mathbb{E}(X)=(1-s)R^* =(1-s)\log\left(\frac{c_u}{c_a}+1\right) $, which is monotonically decreasing on $s$. When the user is under full coverage, the average loss with the payment $T$ is $(1-s)R^* + T \big|_{s = 1} = T$. When the user does not subscribe to an insurance, the average direct loss is $R^*$. Thus, the user has no incentive to insure if the cost under fully coverage is higher than that under no insurance, i.e., $T > R^*$. Moreover, for $T\leq R^*$, the user will choose to insure if the average loss under the given coverage level $s$ is lower than under no insurance, i.e., $ (1-s)R^* + T \leq R^*$. Therefore, we arrive at the following conditions.  
\begin{condition}[Individual Rationality (IR-$u$)]
\label{con:UserT}
The subscription fee must satisfy $T\leq T_{max}:=R^*=\log\left(\frac{c_u}{c_a}+1\right)$, so that the user prefer to subscribe the insurance.
\end{condition}
\begin{condition}[Incentive Compatibility (IC-$u$)]
\label{con:UserS0}
For the subscription fee $T\leq T_{\max}$, the user will subscribe the insurance if the coverage level $s$ satisfies $s\geq s_0=\frac{T}{R^*} = \frac{T}{\log\left(\frac{c_u}{c_a}+1\right) } $.
\end{condition}
Inequalities in Condition \ref{con:UserT} and Condition \ref{con:UserS0} are known as individual rationality (IR-$u$) constraint and incentive compatibility (IC-$u$) constraint, respectively. The user will enroll only when (IR-$u$) and (IC-$u$) constraints are satisfied. Note that when $c_u$ is large and $c_a$ is small, i.e., the saddle-point risk level $R^*$ is high, $T_{\max}$ is large and $s_0(T)$ is small, i.e., when the cost of the user to put local protections is large, and the cost of the attacker to conduct cyber-attack is small, the price of the subscription fee is large, but the minimum coverage is low. Note that $s_0$ is  monotonically increasing on $T$, moreover, when $T=0$, $s=0$, i.e., the user will accept any coverage level when there is no charge for the insurance premium. When $T = T_{\max}$, $s=1$, i.e., the user only accept a full coverage when the subscription fee is the maximum. 
 
The insurer charges a subscription fee $T$ from the user, i.e., the insurer has a gross profit of $T$. However, the insurer also pays the user an average amount of $sR^*=s\log\left(\frac{c_u}{c_a}+1\right)$ from Corollary \ref{rem:Exloss}. Thus, the average operating profit of the insurer is $T - sR^*$, which must be larger than or equal to $0$ so that the insurer will provide the insurance. Thus, we have the following condition. 
\begin{condition}[Individual Rationality (IR-$i$)]
\label{con:InsurerProfit}
The insurer will provide the insurance if $T - sR^*=T - s\log\left(\frac{c_u}{c_a}+1\right) \geq 0$.
\end{condition}
Recall Proposition \ref{pro:Fun}, the insurer will provide the insurance when the user is insurable, i.e., inequality (\ref{eq:SaddlePointFea}) must be satisfied. Thus, we reach the following proposition that indicates the feasible coverage level. 
\begin{condition}[Feasibility (F-$i$)]
\label{con:InsurerS}
The coverage level $s$ is feasible, i.e., the user is insurable, when $s>  1 - \frac{1}{\gamma\log\left(\frac{c_u}{c_a}+1\right)}$. 
\end{condition}

Condition \ref{con:InsurerProfit} and Condition \ref{con:InsurerS} indicate the individual rationality constraint (IR-$i$) and the feasibility constraint (F-$i$) of the insurer, respectively. With the (IR-$u$) and (IC-$u$) constraints for the user and the (IR-$i$) and (F-$i$) constraints for the insurer, the insurer's objective to minimize the average effective loss of the user and maximize the operating profit can be captured using the following optimization problem:
\begin{equation}
\label{eq:InsurerMinICIR}
\begin{array}{l}
\min\limits_{\{0\leq s\leq 1,T\geq 0 \}} J_i(s,T):=\gamma(1-s)\log\left(\frac{c_u}{c_a}+1\right) + c_s(s\log(\frac{c_u}{c_a}+1)-T) \\ \begin{array}{cc}
{\textrm{s.t.}}&{\textrm{(IR-$u$)},\textrm{(IC-$u$)},\textrm{(IR-$i$)},\textrm{(F-$i$)}.}
\end{array}
\end{array}
\end{equation}
Note that the first term of the objective function is the average effective loss of the user under the coverage $s$, as the insurer also aims to reduce the loss of the user from the attacker. Minimizing the second term of the objective function captures the insurer's objective of making profit. Note that parameter $c_s$ indicates the trade-off of a safer user and a larger profit of the insurer. 

Furthermore, the solution of Problem (\ref{eq:InsurerMinICIR}) and the corresponding SPE defined in Definition \ref{def:SaddlePoint} yields an equilibrium for the bi-level game in Case 1 which can be defined as
\begin{definition}
\label{def:EquStrategy}
Let $\mathcal{S}_i$ be the action set for the insurer, $\mathcal{S}_u(s)$ and $\mathcal{S}_a(s)$ be the action sets for the user and the attacker given the insurance coverage level, the strategy pair $(p_u^*,p_a^*,\{s^*,T^*\})$ is called a bi-level game Nash equilibrium (BGNE) of the bi-level game in Case 1 defined by the triple $G_{\textrm{1}}:=\left< \{User,Attacker,Insurer\},\{\mathcal{S}_u(s),\mathcal{S}_a(s),\mathcal{S}_i\}, K,J_i \right>$, if $\{s^*,T^*\}$ solves Problem (\ref{eq:InsurerMinICIR}) with the BGNE objective function $J_i^*$, and the strategy pair $(p_u^*,p_a^*)$ is the SPE of the zero-sum game defined in Definition \ref{def:SaddlePoint} with the SPE objective function $K^*$ under the insurance policy $\{s^*,T^*\}$. 
\end{definition}
Note that the insurer's Problem (\ref{eq:InsurerMinICIR}) is a linear programming problem as the objective function and all the constraints are linear in $s$ and $T$. Instead of solving this problem, we first observe that (IR-$i$) and (IC-$u$) together indicate that the insurance policy $s$ and $T$ must satisfy 
\begin{equation}
\label{eq:IRIICu}
T =sR^*=s\log\left(\frac{c_u}{c_a}+1\right).
\end{equation}
\begin{corollary}
\label{rem:IRIICu}
Equality (\ref{eq:IRIICu}) indicates the following observations:
\begin{itemize}
\item[(i)] {\rm Zero Operating Profit Principle:} The insurer's operating profit is always $0$, as $T-sR^* = 0$.
\item[(ii)] {\rm Linear Insurance Policy Principle:} The insure can only provide the insurance policy $s$ and $T$ that satisfies (\ref{eq:IRIICu}), so that the user subscribes to the insurance and the insurer provides the insurance. 
\end{itemize}
\end{corollary}
Corollary 3 reveals a zero operating profit principle and a linear insurance policy principle for the insurer. These principles hold in Case 2 and 3 as well. With (\ref{eq:IRIICu}), the linear insurance policy indicates that the ratio of the subscription and the coverage level only depends on the saddle-point risk $R^*$, which is determined by the costs seen in Remark \ref{rem:SPRisk}. It provides a fundamental principle for designing the insurance policy.  

With (\ref{eq:IRIICu}), the optimal insurance for the insurer can be summarized using the following proposition. 
\begin{proposition}
\label{pro:OptimalInsurancePolicy}
The optimal insurance policy for the insurer is
\begin{equation}
\label{eq:OptimalInsurancePolicy}
s^*  = 1 \textrm{\ \ \ } T^* = T_{\max} = \log\left(\frac{c_u}{c_a}+1\right). 
\end{equation}
\end{proposition}
Proposition \ref{pro:OptimalInsurancePolicy} shows that a full coverage level and a maximum subscription fee are the optimal insurance policy of the insurer. Together with Proposition \ref{pro:SaddlePointSol}, we have the following proposition of the BGNE of the bi-level game in Case 1. 
\begin{proposition}
\label{pro:EquBi}
The bi-level game of Case 1 admits a unique BGNE solution $(p_u^*,p_a^*,\{s^*,T^*\}) = (0,0,\{1,\log\left(\frac{c_u}{c_a}+1\right)\})$. At the equilibrium, the insurer provides a full coverage for the user and charges a maximum subscription fee from the user. The user and attacker have no incentives to take actions at the equilibrium as the cost would be too high. The equilibrium also demonstrates that cyber insurance will effectively mitigate the loss.
\end{proposition}

%
%
The analysis of the bi-level structure of the game informs the optimal insurance policies to transfer risks from one node to the insurer. The framework can also be extended to a scheme over interdependent infrastructures as illustrated in Fig. \ref{networkinsurance}. The cyber risks at one node can propagate to other nodes when there are cyber, physical, human, and social interdependencies. The insurance of one node can play a important role in well-being of the entire system. We can anticipate that the insurance problem should take into account network effects. This research can also be further extended to investigate the impact of dynamic evolutions of the risks on the mechanism of insurance and the protection behaviors of the agents.

\begin{figure}[]
\centering
\includegraphics[angle=270,width=1.0\textwidth]{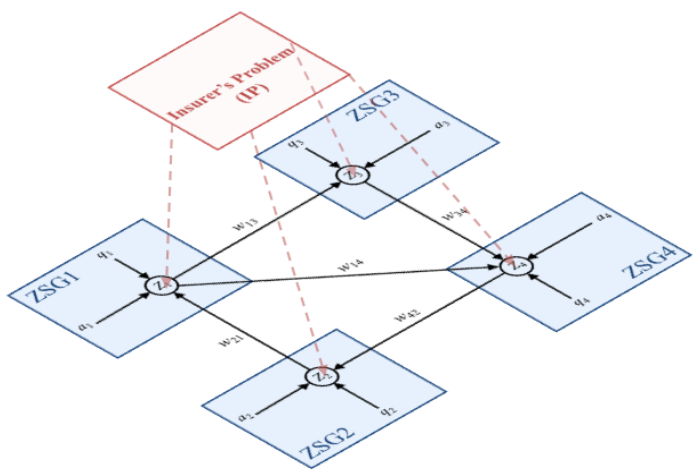}
\vspace{-10mm}
\caption{Cyber insurance over interdependent infrastructures}
\label{networkinsurance}
\end{figure}
\bibliographystyle{spmpsci}      

\bibliography{DraftRZ.bib}

\end{document}